\documentclass[letterpaper, 10pt, conference]{ieeeconf}      

\IEEEoverridecommandlockouts                              

\usepackage{setspace}
\usepackage{graphicx}
\usepackage{epstopdf}
\usepackage{caption}
\usepackage{epsfig} 
\usepackage{amsmath} 
\usepackage{amssymb} 
\usepackage{color}
\usepackage{subcaption}
\usepackage{verbatim}  
\newtheorem{theorem}{Theorem}
\newtheorem{lemma}{Lemma}
\newtheorem{proposition}{Proposition}

\newtheorem{remark}{Remark}

\usepackage[ colorlinks = true,
 linkcolor = blue,
 urlcolor = blue,
 citecolor = red,
 anchorcolor = green,
]{hyperref}

\allowdisplaybreaks[1]

\title{\LARGE \bf 	
	 Hypothesis Testing in the High Privacy Limit
}

\author{Jiachun Liao, Lalitha Sankar,  Vincent Y. F. Tan and Flavio P. Calmon
\thanks{J.~Liao and L.~Sankar are    with the School of Electrical, Computer and
Energy Engineering, Arizona State University, Tempe, AZ  (emails: jiachun.liao@asu.edu, lalithasankar@asu.edu). V.~Y.~F.~Tan is with the Department of Electrical and Computer Engineering and the Department of Mathematics, National University of Singapore (NUS), Singapore (e-mail: {vtan@nus.edu.sg}). F.~du Pin Calmon is with the  IBM T.J.~Watson Research Center in Yorktown Heights, New York (email: fdcalmon@us.ibm.com). } 
\thanks{This work was supported in part by the National Science Foundation under grant CCF\--1350914 and CIF\--1422358.}
}

\begin{document}

\maketitle
\thispagestyle{empty}
\pagestyle{empty}

\begin{abstract}
Binary hypothesis testing under the Neyman-Pearson formalism is a statistical inference framework for distinguishing data generated by two different source distributions. Privacy restrictions may require the curator of the data or the data respondents themselves to share data with the test only after applying a randomizing \textit{privacy mechanism}. Using mutual information as the privacy metric and the relative entropy between the two distributions of the output (post-randomization) source classes as the utility metric (motivated by the Chernoff-Stein Lemma), this work focuses on finding an optimal mechanism that maximizes the chosen utility function while ensuring that the mutual information based leakage for both source distributions is bounded. 
Focusing on the high privacy regime, an Euclidean information-theoretic (E-IT) approximation to the tradeoff problem is presented. It is shown that the solution to the E-IT approximation is independent of the alphabet size and clarifies that a mutual information based privacy metric 
preserves the privacy of the source symbols in inverse proportion to their likelihood.
\end{abstract}

\section{INTRODUCTION}
The use of large datasets to test two or more hypotheses (e.g., the 1\% theory of income distribution in the United States) relies on the classical statistical inference framework of binary (or more generally $M$-ary) hypothesis testing. In particular, binary hypothesis testing under the Neyman-Pearson setup is used to distinguish data generated by two different source classes. However, privacy restrictions may require the curator of the data or the data respondents themselves to share data with the test only after applying a randomizing \textit{privacy mechanism} chosen to ensure that some measure of statistical utility is achieved while simultaneously guaranteeing some measure of privacy. 

In this paper, we study this problem of designing a privacy mechanism when the privacy metric is quantified by the mutual information (average leakage) of a source class and its randomized output and the utility metric is quantified by the relative entropy of the output (post-randomization) distributions for the two source classes. The resulting utility-privacy tradeoff problem involves maximizing the relative entropy subject to constraints on information leakage for both source classes memoryless mappings, i.e., conditional probability matrices between the input and output distributions. The resulting optimization problem involves maximizing a convex function over a convex set which is, in general, NP-hard. We approximate the tradeoff problem in the high privacy regime (near zero leakage) using techniques from Euclidean information theory (E-IT); these techniques have found use in developing capacity results in  \cite{EITzheng2008}, \cite{EIT2015}. Our results show that the solution to the E-IT approximation is independent of the alphabet size and that a mutual information based privacy metric exploits the source statistics and preserves the privacy of the source symbols in inverse proportions to their likelihoods.

{\bf Related work}: The problem of   designing privacy mechanisms for hypothesis testing has recently gained interest. Kairouz \textit{et al.}~\cite{Kairouz2014} show that the optimal locally differential privacy (L-DP) mechanism has a \textit{staircase} form and can be obtained as a solution of a linear program. Li and Ochetering~\cite{Li_Ochetering} considered the problem of an adversarial smart meter data collector interested in learning private behavior of consumers via a hypothesis test with the goal of finding the optimal power consumption policy  at the consumer to limit such inference. Our problem differs from both these efforts in using mutual information as the privacy metric. It is worth noting that the L-DP formulation in \cite{Kairouz2014} by its very definition of requiring the mechanism to limit distinction between any two letters of the source alphabet for a given output is focused on the high privacy regime---the authors exploit a specific sub-linear function of the relative entropy function to simplify the tradeoff problem to a linear program. In contrast, the E-IT approximation we use leads to a convex program for which we present closed-form  solutions. 

\section{NOTATION}
We use bold capital letters to represent matrices, e.g. $\mathbf{X}$ is a matrix with the $i^{\mathrm{th}}$ row (column) indexed as $\mathbf{X}_i$ and the entry at the $i^{\mathrm{th}}$ row $j^{\mathrm{th}}$ column indexed as $X_{ij}$; and use bold lower case letters to represent vectors, e.g. $\mathbf{x}$ is a vector with the $i^{\mathrm{th}}$ entry indexed as $x_i$; and denote sets by capital calligraphic letters, e.g., $\mathcal{X}$.
For vectors $\mathbf{a}$ and $\mathbf{b}$, as well as functions $f$ and $g$,  $\big[\frac{f(\mathbf{a})}{g(\mathbf{b})}\big]$ is a diagonal matrix with the $i^{\mathbf{th}}$ diagonal entry being $\frac{f(a_i)}{g(b_i)}$, e.g., the diagonal matrix $[\frac{\mathbf{a}^2}{\sqrt{\mathbf{b}}}]$ has diagonal entries $\frac{a_i^2}{\sqrt{b_i}}$. We denote the $l_2$-norm of a vector $\mathbf{x}$ by $\|\mathbf{x}\|$. Throughout, probability mass functions are denoted as row vectors, e.g., $\mathbf{p}$. $D$ denotes relative entropy and $I$ denotes mutual information. We can write the mutual information between two random variables or between a probability distribution and the corresponding conditional probability matrix, e.g. for two random variables $X,\hat{X}$ with $X\sim \mathbf{p}$ and $\hat{X}|\{X=x\}\sim \mathbf{P}_{\hat{X}|X=x}$, the mutual information is denoted as $I(X;\hat{X})$ or $I(\mathbf{p},\mathbf{P}_{\hat{X}|X})$.

\section{SYSTEM MODEL}
We consider the binary hypothesis testing framework that distinguishes between two independent and identically distributed (i.i.d.) discrete source classes. Let $X^n=(X_1,X_2,\ldots,X_n)$ denote the sequence of $n$ random variables to be used for hypothesis testing, where entries $X_i\in \mathcal{X}, i \in \{1,2,\ldots,n\}$ are drawn i.i.d.\ according to a probability distribution $\mathbf{p}$, such that its realization $\mathbf{x}^n$ is the available data. The two hypotheses are $H_1:\, \mathbf{p}=\mathbf{p}_1$ and $H_2:\, \mathbf{p}=\mathbf{p}_2$. Let $\mathcal{A}_1^{(n)}$ and $\mathcal{A}_2^{(n)}=\big(\mathcal{A}_1^{(n)}\big)^c$ be the decision regions in $\mathcal{X}^n$, such that from the Neyman-Pearson lemma, $\mathcal{A}_1^{(n)} =\big\{ \mathbf{x}^n \,:\,  \frac{\mathbf{p}_1(\mathbf{x}^n)}{\mathbf{p}_2(\mathbf{x}^n)}>T\big\}$, where $T$ is the threshold for the  likelihood ratio test. Furthermore, let $\beta^{(n)}_1$ and $\beta^{(n)}_2$ be the probabilities of error, such that type-I error $\beta^{(n)}_1$ (resp.\ type-II error $\beta^{(n)}_2$) is the probability of error from choosing $H_2$ (resp. $H_1$) when $\mathbf{x}^n$ is generated by $\mathbf{p}_1$ (resp.\ $\mathbf{p}_2$). From the Chernoff-Stein lemma~\cite[Chap.~11]{IT_Cover}, for a desired $\beta^{(n)}_1<\delta$ (where $\delta>0$), i.e., bounding the probability of declaring that the test sequence $\mathbf{x}^n$ is from $\mathbf{p}_2$ when $\mathbf{p}_1$ is true, the type-II error exponent of minimal $\beta^{(n)}_2$, denoted as $\beta^{(n)}_2(\delta)$, is $\lim_{n\rightarrow \infty} -\frac{1}{n}\log\beta^{(n)}_2(\delta) =D(\mathbf{p}_1 \| \mathbf{p}_2)$ independent of $\delta$.

In most data collection and classification applications, there may be an additional requirement to ensure that the data while providing utility does not leak information about the respondents of the data. This in turn implies that the data provided to the hypothesis test is not the same as the original data, but instead a randomized version that ensures some measures of privacy (information leakage) and utility. Specifically, we use mutual information as a measure of a stochastic information leakage between the input sequence and the randomized output sequence used by the test. The goal is to find the randomizing mapping, henceforth referred to a \textit{privacy mechanism}, such that a measure of utility of the data is maximized while ensuring that the mutual information based leakages for all source classes are bounded.

As mentioned earlier, we assume that the source distributions are i.i.d.; furthermore, modeling the large dataset problem, we assume that the length of the data sequence $n$ is large. In the absence of a \textit{privacy mechanism}, asymptotic results of the binary hypothesis testing problem exploit  the i.i.d.\ nature of the sources. Preserving the utility of the data usually requires preserving its i.i.d.\ property, and thus, we assume that the randomizing \textit{privacy mechanism} for the hypothesis testing problem considered here is an i.i.d.\ mechanism. Let $\mathbf{W}$, an $M\times N$ conditional probability matrix, denote the \textit{privacy mechanism} which maps the $M$ letters of the input alphabet $\mathcal{X}$ to $N \leq M$ letters of the output alphabet $\hat{\mathcal{X}}$. The assumption $N\leq M$ implies that we map the input probability distribution to a simplex with at most the same cardinality. Thus, the i.i.d.\ sequence $X^n \sim \mathbf{p}_k, k\in \{1,2\}$, is mapped to an output sequence $\hat{X}^n$ whose entries $\hat{X}_j \in \hat{\mathcal{X}}$ for all $j\in \{1,\ldots,N\}$ are i.i.d.\ with the distribution $\mathbf{p}_k\mathbf{W}$. Thus the hypothesis test is now performed on a sequence $\hat{X}^n$ that belongs to one of two source classes with distributions\footnote{We remind that the distribution $\mathbf{p}_1\mathbf{W}$ is the {\em output distribution} induced by the input $\mathbf{p}$ and the privacy mechanism (transition matrix) $\mathbf{W}$.} $\mathbf{p}_1\mathbf{W}$ and $\mathbf{p}_2\mathbf{W}$, respectively. For the two i.i.d.\ source classes, the type-II error exponent of the test now is $\lim_{n\rightarrow \infty} -\frac{1}{n}\log\beta^{(n)}_2(\delta) =D(\mathbf{p}_1\mathbf{W} \| \mathbf{p}_2\mathbf{W})$.

 To design an appropriate \textit{privacy mechanism}, we wish to maximize the error exponent $D(\mathbf{p}_1\mathbf{W} \| \mathbf{p}_2\mathbf{W})$ subject to the following leakage constraints: $I(\mathbf{p}_1,\mathbf{W})\leq \epsilon_1$ and $I(\mathbf{p}_2,\mathbf{W})\leq \epsilon_2$. Formally, the utility-privacy tradeoff problem is that finding the optimal \textit{privacy mechanism} $\mathbf{W}^*$, i.e., the optimal solution of the following optimization problem:
 \begin{equation}\label{eq:OrigProHP}
 	\begin{aligned}
 		\max_{\substack{\mathbf{W}}}\quad & D(\mathbf{p}_1\mathbf{W} \| \mathbf{p}_2\mathbf{W})\\		
 		\mathrm{s.t.} \quad & I(\mathbf{p}_1,\mathbf{W})\leq \epsilon_1\\
 		& I(\mathbf{p}_2,\mathbf{W})\leq \epsilon_2
 	\end{aligned}
 \end{equation}
 where $\epsilon_1 \in [0, H(\mathbf{p}_1)] $ and $\epsilon_2 \in [0, H(\mathbf{p}_2)]$ are the permissible upper  bounds of the privacy leakages for the two source classes. The problem in \eqref{eq:OrigProHP} maximizes a convex objective function over a convex feasible region, thus, its optimal solutions are on the boundary of the feasible region. When the feasible region is a polytope,  closed-form  expressions for the optimal solution and objective can potentially be computed. However, the region bounded by the two mutual information constraints is not a polytope. While there exist computational methods to obtain a solution that approximate the feasible region by an intersection of polytopes \cite{MinConcave}, our focus is on developing a meaningful approximation for the problem in~\eqref{eq:OrigProHP} in specific privacy or utility regimes. 
 
 To this end, we focus on developing approximations of \eqref{eq:OrigProHP}. One possible regime to use an approximation approach is the high privacy regime, i.e., $0\leq\epsilon_k\ll \min(H(\mathbf{p}_1),H(\mathbf{p}_2))$ for all $k\in\{1,2\}$. In this regime, one can use Taylor series expansions to approximate both the objective function and the constraints. Such approximations to the relative entropy and the mutual information were first considered in \cite[Theorem 4.1]{ITSt_Tutorial}. More recently, analyses based on such approximations, referred to as Euclidean Information Theory (E-IT), have been found to be useful in a variety of graphical model learning~\cite{Tan11} and network information theory problems \cite{EITzheng2008}\cite{EIT2015}. In the following section, we detail the approximation problem and the corresponding closed-form  solution.

\section{MAIN RESULTS}\label{section:main_results}
In this section, we summarize all main results. Our first result captures the mapping of the utility-privacy tradeoff problem in \eqref{eq:OrigProHP} to an Euclidean information-theoretic (E-IT) approximation problem. Our second result summarizes a closed-form  result for the E-IT approximation problem while our last result introduces the approximate solution for \eqref{eq:OrigProHP} based on the solution of the E-IT approximation problem.\\
Without loss of generality, all our results are presented for the case of $|\mathcal{\hat{X}}|=M$, i.e., the size of output alphabet is the same as the input alphabet. We highlight in the sequel how the method can be used for $|\mathcal{\hat{X}}|=N<M$.

\begin{proposition}\label{Proposition:approximation}
	In the high privacy regime with $0\leq\epsilon_k \ll \min(H(\mathbf{p}_1),H(\mathbf{p}_2)),\, k \in \{1,2\}$, the \textit{privacy mechanism} $\mathbf{W}$ is chosen as a perturbation of a perfect privacy $(\epsilon_1=\epsilon_2=0)$ achieving mechanism $\mathbf{W}_0$, i.e., $\mathbf{W}=\mathbf{W}_0+\boldsymbol{\Theta}$. The mechanism $\mathbf{W}_0$ is a rank-1 matrix with every row being equal to a row vector $\mathbf{w}_0$ whose entries $w_{0j}$ satisfy $\sum_{j=1}^{M}w_{0j}=1$ and $w_{0j}>0$, for all $j\in\{1,\ldots,M\}$. The perturbation matrix $\boldsymbol{\Theta}$ is an $M\times M$ matrix with entries $\Theta_{ij}$ satisfying $\sum_{j=1}^{M}\Theta_{ij}=0$ and $|\Theta_{ij}| \leq \rho w_{0j}$, for all $i,j\in \{1, \ldots, M\}$. For this perturbation model,
	the utility-privacy tradeoff problem in \eqref{eq:OrigProHP} can be approximated as 
	\begin{equation}\label{eq:OrigProHP_approx}
	\begin{aligned}
	\max_{\substack{\boldsymbol{\Theta}}}\quad &\frac{1}{2}\big\|(\mathbf{p}_1-\mathbf{p}_2)\boldsymbol{\Theta}[(\mathbf{w}_0)^{-\frac{1}{2}}]\big\|^2	\\
	\mathrm{s.t.} \quad &\frac{1}{2}\sum_{i=1}^{M}p_{1i}\big\|\boldsymbol{\Theta}_i[(\mathbf{w}_0)^{-\frac{1}{2}}]\big\|^2 \leq \epsilon_1;\\
	&\frac{1}{2}\sum_{i=1}^{M}p_{2i}\big\|\boldsymbol{\Theta}_i[(\mathbf{w}_0)^{-\frac{1}{2}}]\big\|^2 \leq \epsilon_2.\\
	& \sum_{j=1}^{M}\Theta_{ij}=0, \quad i\in \{1, \ldots, M\}
	\end{aligned}
	\end{equation}
	where $p_{ki}$, for $k=1,2$, is the $i^{\mathrm{th}}$ entry of $\mathbf{p}_k$, $\boldsymbol{\Theta}_i$ is the $i^{\mathrm{th}}$ row of $\boldsymbol{\Theta}$, and $[(\mathbf{w}_0)^{-\frac{1}{2}}]$ is a diagonal matrix with $i^{\mathrm{th}}$ diagonal entry, for all $i$, being $(w_{0i})^{-\frac{1}{2}}$. For ease of analysis, setting $\mathbf{A}=\boldsymbol{\Theta}[(\mathbf{w}_0)^{-\frac{1}{2}}]$, \eqref{eq:OrigProHP_approx} can be rewritten as 
	\begin{equation}\label{eq:OrigProHP_approx_1}
	\begin{aligned}
	\max_{\substack{\mathbf{A}}}\quad & \frac{1}{2}(\mathbf{p}_1-\mathbf{p}_2)\mathbf{A}\mathbf{A}^T(\mathbf{p}_1-\mathbf{p}_2)^T\\
	\mathrm{s.t.} \quad &\frac{1}{2}\sum_{i=1}^{M}p_{1i}\mathbf{A}_i\mathbf{A}_i^T\leq \epsilon_1\\
	&\frac{1}{2}\sum_{i=1}^{M}p_{2i}\mathbf{A}_i\mathbf{A}_i^T \leq \epsilon_2\\
	& \mathbf{A}(\sqrt{\mathbf{w}_0})^T=\mathbf{0}.
	\end{aligned}
	\end{equation}	
\end{proposition}
\begin{remark}
 The approximation problems in \eqref{eq:OrigProHP_approx} and \eqref{eq:OrigProHP_approx_1} result from the observation that both the relative entropy and mutual information can be approximated by the $\chi^2$ divergence in the high privacy (low leakage) regime. Thus, the problem simplifies to maximizing a quadratic objective function with quadratic constraints, whose optimal solution is a perturbation matrix $\boldsymbol{\Theta}^*=\mathbf{A}^*\big[(\sqrt{\mathbf{w}_0})\big]$ and  $\mathbf{A}^*$ optimizes the problem in \eqref{eq:OrigProHP_approx_1}.
\end{remark}

\begin{theorem}\label{Theorem:approximation_simplification}
	The optimization problem in \eqref{eq:OrigProHP_approx_1} reduces to one with a vector variable $\boldsymbol{\alpha} \in \mathbb{R}^{M \times 1}$ as
	\begin{equation}\label{eq:OrigProHP_approx_2}
	\begin{aligned}
	\max_{\substack{\boldsymbol{\alpha}}}\quad & 
	\frac{1}{2}\boldsymbol{\alpha}(\mathbf{p}_1-\mathbf{p}_2)^T(\mathbf{p}_1-\mathbf{p}_2)\boldsymbol{\alpha}^T\\
	\mathrm{s.t.} \quad &\frac{1}{2}\boldsymbol{\alpha}[\mathbf{p}_1]\boldsymbol{\alpha}^T\leq \epsilon_1\\
	&\frac{1}{2}\boldsymbol{\alpha}[\mathbf{p}_2]\boldsymbol{\alpha}^T \leq \epsilon_2,
	\end{aligned}
	\end{equation}
	where the absolute value of the $i^{\mathrm{th}}$ entry $\alpha_i$ of $\boldsymbol{\alpha}$, for all $i \in \{1,..,M\}$, is the Euclidean norm of the $i^{\mathrm{th}}$ row $\mathbf{A}_i$ of $\mathbf{A}$. The $M\times M$ matrix $\mathbf{A}^*$ optimizing \eqref{eq:OrigProHP_approx_1} is obtained from the optimal solution $\boldsymbol{\alpha}^*$ of \eqref{eq:OrigProHP_approx_2} as 
	a rank-1 matrix whose $i^{\mathrm{th}}$ row, for all $i$, is given by $\alpha_i^*\mathbf{v}$ where $\alpha_i^*$ is the $i^{\mathrm{th}}$ entry of  $\boldsymbol{\alpha}^*$, and $\mathbf{v}$ is a unit norm vector that is orthogonal to $\sqrt{\mathbf{w}_0}$, such that
	\begin{align}
	\label{eq:lemma1_Aopt}
	&\mathbf{A}^*=(\boldsymbol{\alpha}^*)^T\mathbf{v}\\
	\label{eq:lemma1_v}
	&\mathbf{v}(\sqrt{\mathbf{w}_0})^T=0\\
	\label{eq:lemma1_vnorm}
	&\|\mathbf{v}\|=1.
	\end{align}    
\end{theorem}
\begin{remark}
	Theorem \ref{Theorem:approximation_simplification} results from exploiting the structure of the QCQP problem in \eqref{eq:OrigProHP_approx_1} to further restrict the matrix $\mathbf{A}$ to a vector $\boldsymbol{\alpha}$ as shown in the sequel.
\end{remark}

\begin{theorem}\label{Theorem:optimalsol_HPapprox}
	An optimal \textit{privacy mechanism} $\mathbf{W}'$ for the E-IT approximation problem in \eqref{eq:OrigProHP_approx} is
	\begin{align}
	\label{eq:W_generation}
	\mathbf{W}'=\mathbf{W}_0+(\boldsymbol{\alpha}^*)^T\mathbf{v}\cdot\big[\sqrt{\mathbf{w}_0}\big]
	\end{align}
	where $\mathbf{W}_0$ is given by Proposition \ref{Proposition:approximation}, $\mathbf{v}$ is chosen to satisfy \eqref{eq:lemma1_v} and \eqref{eq:lemma1_vnorm}, and for $\lambda^*=\|\mathbf{p}_1-\mathbf{p}_2\|^2$ and $\mathbf{v}^*=\frac{\mathbf{p}_1-\mathbf{p}_2}{\|\mathbf{p}_1-\mathbf{p}_2\|}$ being the eigenvalue and eigenvector, respectively, of $(\mathbf{p}_1-\mathbf{p}_2)^T(\mathbf{p}_1-\mathbf{p}_2)$, the optimal solution, $\boldsymbol{\alpha}^*$, of \eqref{eq:OrigProHP_approx_2} is given as:
		\begin{enumerate}
			\item if only the first constraint in \eqref{eq:OrigProHP_approx_2} is active,
			\begin{align}
			\label{eq:const1_active}
			\frac{\mathbf{v}^*\Big[\frac{\mathbf{p}_2}{(\mathbf{p}_1)^2}\Big](\mathbf{v}^*)^T}{\mathbf{v}^*\big[(\mathbf{p}_1)^{-1}\big](\mathbf{v}^*)^T} < \frac{\epsilon_2}{\epsilon_1},
			\end{align}
			and the optimal solution $\boldsymbol{\alpha}^*$ is
			\begin{align}
			\label{eq:optalpha_const1}
			\boldsymbol{\alpha}^*=\pm\sqrt{\frac{2\epsilon_1}{\mathbf{v}^*\big[(\mathbf{p}_1)^{-1}\big](\mathbf{v}^*)^T}}\mathbf{v}^*\big[(\mathbf{p}_1)^{-1}\big];
			\end{align}
			\item if only the second constraint in \eqref{eq:OrigProHP_approx_2} is active,
			\begin{align}
			\label{eq:const2_active}
			\frac{\mathbf{v}^*\Big[\frac{\mathbf{p}_1}{(\mathbf{p}_2)^2}\Big](\mathbf{v}^*)^T}{\mathbf{v}^*\big[(\mathbf{p}_2)^{-1}\big](\mathbf{v}^*)^T} < \frac{\epsilon_1}{\epsilon_2},
			\end{align}
			and the optimal solution $\boldsymbol{\alpha}^*$ is
			\begin{align}
			\label{eq:optalpha_const2}
			\boldsymbol{\alpha}^*=\pm\sqrt{\frac{2\epsilon_2}{\mathbf{v}^*\big[(\mathbf{p}_2)^{-1}\big](\mathbf{v}^*)^T}}\mathbf{v}^*\big[(\mathbf{p}_2)^{-1}\big];
			\end{align}	
			\item when both constraints in \eqref{eq:OrigProHP_approx_2} are active, the optimal solution $\boldsymbol{\alpha}^*$ is
			\begin{align}
			\label{eq:optimalalpha}
			\boldsymbol{\alpha}^*=\pm\frac{\lambda^*}{2}\mathbf{v}^*\Big(\eta_1^*\big[\mathbf{p}_1\big]+\eta_2^*\big[\mathbf{p}_2\big]\Big)^{-1}
			\end{align}
			where $\eta_1^*>0$ and $\eta_2^*>0$ satisfy
			\begin{align}
			\label{eq:eta1} \mathbf{v}^*\big[\mathbf{p}_1\big]\Big(\eta_1^*\big[\mathbf{p}_1\big]+\eta_2^*\big[\mathbf{p}_2\big]\Big)^{-2}(\mathbf{v}^*)^T&=\frac{8\epsilon_1}{(\lambda^*)^2}\\
			\label{eq:eta2}
			\mathbf{v}^*\big[\mathbf{p}_2\big]\Big(\eta_1^*\big[\mathbf{p}_1\big]+\eta_2^*\big[\mathbf{p}_2\big]\Big)^{-2}(\mathbf{v}^*)^T&=\frac{8\epsilon_2}{(\lambda^*)^2}
			\end{align}
		\end{enumerate}	
\end{theorem}
\begin{remark}

\begin{figure}
\centering
\includegraphics[width = \columnwidth]{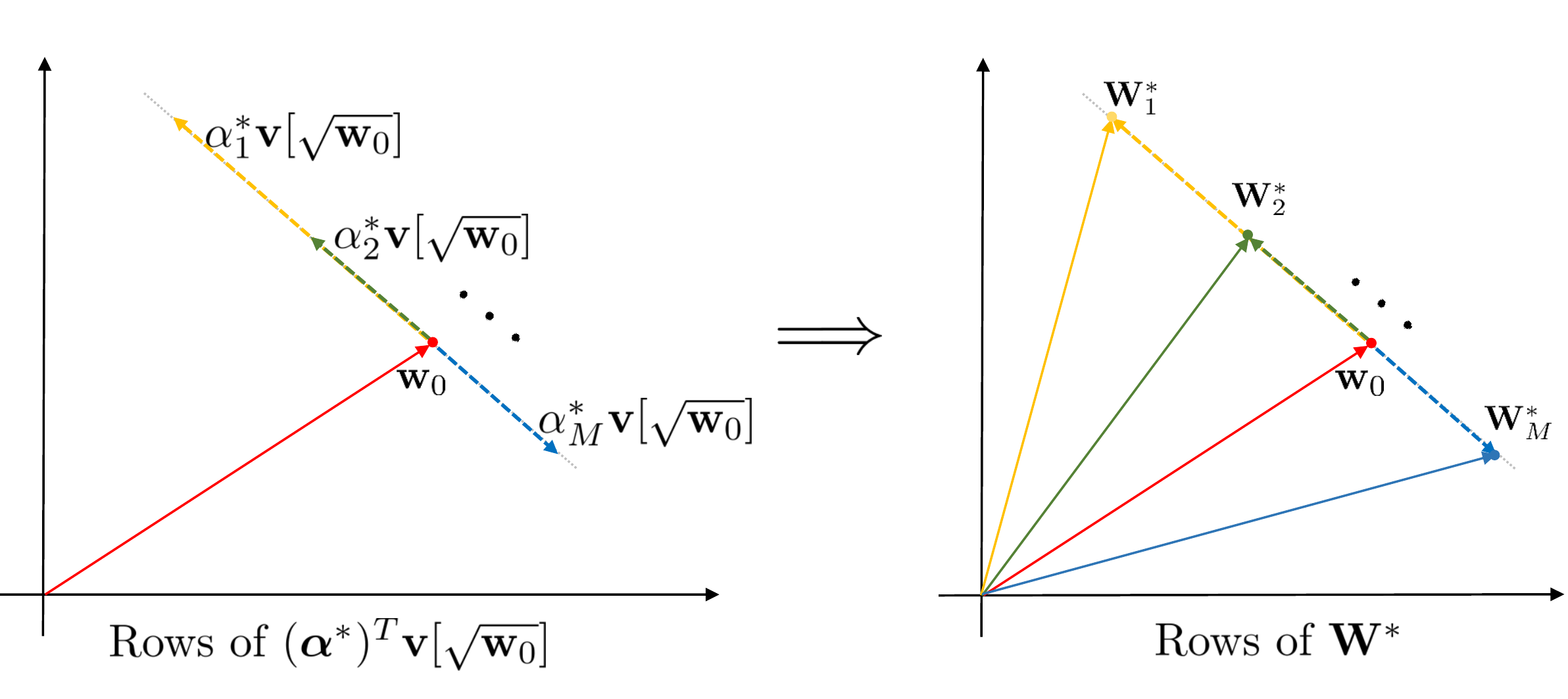}
\caption{Illustration of  \eqref{eq:W_generation}. The ``nominal'' row vector $\mathbf{w}_0$ is perturbed in different directions depending on $\boldsymbol{\alpha}^*$ and $\mathbf{v}$.}
\label{fig:interp}
\end{figure}
	Note that the optimal mechanism $\mathbf{W}'=\mathbf{W}_0+(\boldsymbol{\alpha}^*)^T\mathbf{v}\big[\sqrt{\mathbf{w}_0}\big]$, via $\boldsymbol{\alpha}^*$, captures the fact that a statistical privacy metric, such as mutual information, takes into consideration the source distribution in designing the (perturbation $\boldsymbol{\Theta}^*$) mechanism. In fact, the solutions for $\boldsymbol{\alpha}^*$ in \eqref{eq:optalpha_const1}, \eqref{eq:optalpha_const2} and \eqref{eq:optimalalpha} capture this through the term $\big(\eta_1^*[\mathbf{p}_1]+\eta_2^*[\mathbf{p}_2]\big)^{-1}$. The vector $\mathbf{v}^*$ indicates the direction along which the objective function, i.e., the relative entropy, grows fastest. See Fig.~\ref{fig:interp} for an illustration of the result in Theorem \ref{Theorem:optimalsol_HPapprox}.
	Note that from \eqref{eq:optalpha_const1}, \eqref{eq:optalpha_const2} and \eqref{eq:optimalalpha} for a uniformly distributed source, i.e., the entries of $\mathbf{p}_1$ or $\mathbf{p}_2$ are $\frac{1}{M}$, all entries
	of $\mathbf{v}^*$ have the same scaling such that $\boldsymbol{\alpha}^*$ is in the direction of $\mathbf{v}^*$. However, for a source with non\--uniform distributed samples, the samples with low probabilities affect the direction of $\mathbf{v}^*\big(\eta_1^*[\mathbf{p}_1]+\eta_2^*[\mathbf{p}_2]\big)^{-1}$ the most. This is a direct consequence of the statistical leakage metric which causes the optimal mechanism to minimize information leakage by perturbing the low probability (more informative) samples proportionately more relative to the higher probability samples. 
\end{remark}

\section{PROOFS OF THE MAIN RESULTS}
In this section, we develop the E-IT approximation of the utility-privacy tradeoff problem in \eqref{eq:OrigProHP} and present the proofs for Proposition \ref{Proposition:approximation}, Theorem \ref{Theorem:approximation_simplification} and Theorem \ref{Theorem:optimalsol_HPapprox} in Section \ref{section:main_results}. 

To develop an approximation using Taylor series, we select an operating point which will be perturbed in order to provide an approximately optimal privacy mechanism. Since our focus is on the high privacy regime $\epsilon_1,\epsilon_2\in [0,\epsilon^*], \epsilon^* \ll \min(H(\mathbf{p}_1),H(\mathbf{p}_2))$, we present the approximation around a perfect privacy operation point, i.e., a \textit{privacy mechanism} $\mathbf{W}_0$ that achieves $\epsilon_1=\epsilon_2=0$ leakage. The following lemma summarizes the set of all mechanisms that achieve perfect privacy. 

\begin{lemma}\label{lemma:perfect_privacy}
	For perfect privacy, i.e., $\epsilon_1=\epsilon_2=0$, the \textit{privacy mechanism} $\mathbf{W}_0$ is a rank-1 matrix with every row being equal to a row vector $\mathbf{w}_0$ where $\mathbf{w}_0$ belongs to probability simplex, such that the entries $w_{0j},j\in \{1,2,\ldots,N\}$ of the vector $\mathbf{w}_0$ satisfy
	\begin{align}
	&\sum_{j=1}^{N}w_{0j}=1,\qquad \mathrm{and}\\
	&w_{0j} \geq 0 \quad \forall j\in \{1, \ldots, N\} .
	\end{align}
\end{lemma}
\begin{proof}
	For any probability distribution $\mathbf{p}$ with entries $p_i, \,i\in \{1,\ldots,M\}$, and a \textit{privacy mechanism} $\mathbf{W}$, the mutual information $I(\mathbf{p},\mathbf{W})$ can be written as
	\begin{align}
 I(\mathbf{p},\mathbf{W})
	&= \sum_{i=1}^{M}\sum_{j=1}^{N}p_iW_{ij}\log\frac{W_{ij}}{\sum_{i=1}^{M}p_iW_{ij}}\\
	\label{eq:perfectprivacy_logInequality}
&	\geq  \sum_{i=1}^{M}p_i\bigg(\sum_{j=1}^{N}W_{ij}\bigg)\log\frac{\sum_{j=1}^{N}W_{ij}}{\sum_{j=1}^{N}\sum_{i=1}^{M}p_iW_{ij}}\\
	&= \sum_{i=1}^{M}p_i\sum_{j=1}^{N}W_{ij}\log\frac{1}{1}=0
	\end{align}
	where \eqref{eq:perfectprivacy_logInequality} results from the log-sum inequality. Equality in \eqref{eq:perfectprivacy_logInequality} holds if and only if \cite[Theorem 2.7.1]{IT_Cover}
	\begin{align}
	\label{eq:perfectprivacy_W}
	W_{ij}=\sum_{k=1}^{M}p_kW_{kj},\;\; i\in \{1, \ldots, M\} ,j\in \{1, \ldots, N\}.
	\end{align}
	In other words, perfect privacy, i.e., zero leakage, is achieved when every row of the optimal mechanism $\mathbf{W}_0$ is the same and is equal to the probability distribution $\mathbf{w}_0=\mathbf{p}\mathbf{W}_0$. 
\end{proof}
Thus, for the perfect privacy setting,
the optimal mechanism satisfying \eqref{eq:perfectprivacy_W} is a perfectly noisy mechanism, such that the input distribution has no effect on the output. 
\begin{remark}
	Note that, for any $\mathbf{W}_0$ satisfying \eqref{eq:perfectprivacy_W} that achieves perfect privacy, the utility is $D(\mathbf{p}_1\mathbf{W}_0 \| \mathbf{p}_2\mathbf{W}_0)=0$. Furthermore, the rows of $\mathbf{W}_0$, i.e., $\mathbf{w}_0$, can take any value in an $M$\--- dimensional probability simplex, i.e., $\mathbf{w}_0$ can have any $k\in \{1,\ldots,M\}$ non-zero columns. This implies that perfect privacy, and hence, zero utility, can be achieved for any cardinality of the output from $1$ to $M$.
\end{remark}

\subsection{Proof of Proposition \ref{Proposition:approximation} }
Proposition \ref{Proposition:approximation} presents an E-IT approximation of the utility-privacy tradeoff problem \eqref{eq:OrigProHP} in the high privacy regime where the \textit{privacy mechanism} $\mathbf{W}$ is modeled in \eqref{eq:Theta_to_W}-\eqref{eq:Con_Theta2}. The proof of the proposition is as follows. 

\begin{proof}
	Consider the high privacy regime in which $0\leq\epsilon_k \ll \min(H(\mathbf{p}_1),H(\mathbf{p}_2)),\, k \in \{1,2\}$. In this regime, for a feasible $(\epsilon_1,\epsilon_2)$ pair the \textit{privacy mechanism} $\mathbf{W}$ can be written as a perturbation of $\mathbf{W}_0$ via
	\begin{align}
	\label{eq:Theta_to_W}
	\mathbf{W}=\mathbf{W}_0+\boldsymbol{\Theta}
	\end{align}
	where $\mathbf{W}_0$ is a mechanism achieving perfect privacy (i.e., satisfies \eqref{eq:perfectprivacy_W}) with all rows equal to $\mathbf{w}_0$, where $\mathbf{w}_0$ is chosen such that its entries $w_{0j}\neq 0, \forall j \in \{1,\ldots,N\}$, and $\mathbf{\Theta}$ is a matrix with
	\begin{align}
	\label{eq:Con_Theta1}
	 \sum_{j=1}^{N}\Theta_{ij}&=0,   \qquad \,\,\forall\,  i\in \{1, \ldots, M\} \\
	\label{eq:Con_Theta2}
	 |\Theta_{ij}&| \leq \rho w_{0j}  \quad \forall\,  i\in \{1, \ldots, M\} ,j\in \{1, \ldots, N\} .
	\end{align}
	where the radius  of the neighborhood around $\mathbf{w}_0$ is~$\rho \!\in\! [0,1)$.
	
	Note that \eqref{eq:Con_Theta1} results from the requirement that $\mathbf{W}$ is stochastic, i.e., the rows of $\mathbf{W}$ sum to 1. Since the rows of $\mathbf{W}_0$ also sum to 1, we obtain \eqref{eq:Con_Theta1}. The constraint in \eqref{eq:Con_Theta2} captures the fact that approximating about a perfect privacy achieving mechanism   requires restricting the entries of the perturbation matrix $\boldsymbol{\Theta}$ to be within a fraction $\rho$ of~$\mathbf{w}_0$.
	
	The perturbation modeled in \eqref{eq:Theta_to_W}-\eqref{eq:Con_Theta2} implies that every row in $\mathbf{W}$ as well as the output distribution $\mathbf{p}_1\mathbf{W}$ and $\mathbf{p}_2\mathbf{W}$ are in a neighborhood about $\mathbf{w}_0$ given by
	\begin{subequations}
		\begin{align}
		\label{eq:OrigProbHP_ApproxCond3}
		&\Big|W_{ij}-w_{0j}\Big|\leq\rho w_{0j}, \\
		\label{eq:OrigProbHP_ApproxCond1}
		&\Big|\big(\mathbf{p}_1\mathbf{W}\big)_j-w_{0j}\Big|\leq\rho w_{0j},  \\
		\label{eq:OrigProbHP_ApproxCond2}
		&\Big|\big(\mathbf{p}_2\mathbf{W}\big)_j-w_{0j}\Big|\leq\rho w_{0j},
		\end{align}
	\end{subequations}
	for all $i \in \{1, \ldots, M\}$ and $j \in\{1, \ldots, N\}$.  
In this neighborhood, we can approximate the relative entropy using Taylor series about $\mathbf{W}_0$ as
\begin{align}
& D(\mathbf{p}_1\mathbf{W} \| \mathbf{p}_2\mathbf{W}) \nonumber
\label{eq:ApproximateD_step1} \\ &= \frac{1}{2}\sum_{j=1}^{N}\frac{\big(\sum_{i=1}^{M}p_{1i}\Theta_{ij}-\sum_{i=1}^{M}p_{2i}\Theta_{ij}\big)^2}{\sum_{i=1}^{M}p_{1i}W_{ij}}\nonumber\\*
&\qquad +o\bigg(\Big\|(\mathbf{p}_1-\mathbf{p}_2)\boldsymbol{\Theta}\big[(\mathbf{p}_1\mathbf{W})^{-\frac{1}{2}}\big]\Big\|^2_\infty\bigg)\\
\label{eq:ApproximateD_step2}
&\approx   \frac{1}{2}\Big\|(\mathbf{p}_1-\mathbf{p}_2)\boldsymbol{\Theta}\big[(\mathbf{p}_1\mathbf{W})^{-\frac{1}{2}}\big]\Big\|^2\\
\label{eq:ApproximateD_step3}
&\approx  \frac{1}{2}\big\|(\mathbf{p}_1-\mathbf{p}_2)\boldsymbol{\Theta}[(\mathbf{w}_0)^{-\frac{1}{2}}]\big\|^2,
\end{align}
where \eqref{eq:ApproximateD_step2} results from approximating the relative entropy by the $\chi^2$-divergence term \cite[Theorem 4.1]{ITSt_Tutorial}, and \eqref{eq:ApproximateD_step3} results from applying the neighborhood condition in \eqref{eq:OrigProbHP_ApproxCond1}. 

Similarly, one can approximate the mutual information between the source class $k\in\{1,2\}$ and its output as 
\begin{align}
& I(\mathbf{p}_k,\mathbf{W}) \nonumber\\
\label{eq:ApproximateI_step1}
&=\frac{1}{2}\sum_{i=1}^{M}p_{ki}\bigg(\sum_{j=1}^{N}\frac{\big(W_{ij}-\sum_{i=1}^{M}p_{ki}W_{ij}\big)^2}{W_{ij}} \nonumber\\
&\qquad  + o\Big(\Big\|\big(\mathbf{W}_i-\mathbf{p}_k\mathbf{W}\big)\big[(\mathbf{W}_i)^{-\frac{1}{2}}\big]\Big\|^2_\infty\Big)\bigg)\\
\label{eq:ApproximateI_step2}
&\approx  \frac{1}{2}\sum_{i=1}^{M}p_{ki}\Big\|\big(\mathbf{W}_i-\mathbf{p}_k\mathbf{W}\big)\big[(\mathbf{W}_i)^{-\frac{1}{2}}\big]\Big\|^2 \\
\label{eq:ApproximateI_step3}
&\approx \frac{1}{2}\sum_{i=1}^{M}p_{ki}\big\|\big(\mathbf{w}_0+\boldsymbol{\Theta}_{i}-\mathbf{p}_k\mathbf{W}\big)[(\mathbf{w}_0)^{-\frac{1}{2}}]\big\|^2 \\
\label{eq:ApproximateI_step4}
&\approx \frac{1}{2}\sum_{i=1}^{M}p_{ki}\big\|\boldsymbol{\Theta}_{i}[(\mathbf{w}_0)^{-\frac{1}{2}}]\big\|^2,
\end{align}
where $\mathbf{W}_{i}$ and $\boldsymbol{\Theta}_{i}$ are the $i^{\mathrm{th}}$ rows of $\mathbf{W}$ and $\boldsymbol{\Theta}$, respectively.  With these E-IT approximations of the relative entropy and the mutual information, the utility-privacy tradeoff problem in \eqref{eq:OrigProHP} is approximated as \eqref{eq:OrigProHP_approx}. Let $\mathbf{A}$ be a matrix with entries
 \begin{align}
 \label{eq:A_to_Theta}
 A_{ij}=\frac{\Theta_{ij}}{\sqrt{w_{0j}}},\qquad  i\in \{1, \ldots, M\} ,j \in \{1, \ldots, N\} ,
 \end{align} 
 From \eqref{eq:Con_Theta1}, we have that $\mathbf{A}(\sqrt{\mathbf{w}_0})^T=\mathbf{0}$, where $\sqrt{\mathbf{w}_0}$ is the vector whose entries are the square root of the entries of $\mathbf{w}_0$. Thus,  the approximation in~\eqref{eq:OrigProHP_approx} leads to~\eqref{eq:OrigProHP_approx_1}.
\end{proof}
\begin{remark}
The  row vector $\mathbf{w}_0$ can have any $k\in\{1,\ldots,M\}$ non-zero columns, i.e., it can be any point in an $M$\---dimensional simplex including extremal/corner points. From \eqref{eq:Con_Theta2}, this implies that the columns of $\boldsymbol{\Theta}$ that are non-zero are those for which $w_{0j}$ are non-zero. Without loss of generality, we assume that $M=N$ and that $\mathbf{w}_0$ is chosen in the interior of the simplex. The results apply for any $N<M$.
\end{remark} 

\subsection{Proof of Theorem \ref{Theorem:approximation_simplification}}
Note that in the optimization problem \eqref{eq:OrigProHP_approx_1} the first two constraints only depend on the Euclidean norm of the rows of $\mathbf{A}$. On the other hand, $\mathbf{A}\mathbf{A}^T$ in the objective function includes all possible inner products of any two rows of $\mathbf{A}$. Based on the special structure, we can simplify the problem in \eqref{eq:OrigProHP_approx_1} to obtain the form \eqref{eq:OrigProHP_approx_2} in Theorem \ref{Theorem:approximation_simplification}. 

\begin{proof}
	Consider the following optimization problem obtained from \eqref{eq:OrigProHP_approx_1} without the constraint $\mathbf{A}(\sqrt{\mathbf{w}_0})^T=\mathbf{0}$:
	\begin{equation}\label{eq:OrigProHP_approx_1a}
	\begin{aligned}
	\max_{\substack{\mathbf{A}}}\quad & \sum_{i,j=1}^{M}\frac{1}{2}(\mathbf{p}_1-\mathbf{p}_2)_i(\mathbf{p}_1-\mathbf{p}_2)_j\mathbf{A}_i\mathbf{A}_j^T\\
	\mathrm{s.t.} \quad &\frac{1}{2}\sum_{i=1}^{M}p_{ki}\mathbf{A}_i\mathbf{A}_i^T\leq \epsilon_k \quad k=1,2\\
	\end{aligned}
	\end{equation}	
	Let $\Delta_i \triangleq (\mathbf{p}_1-\mathbf{p}_2)_i$, $i\in\{1,\ldots,M\}$. Furthermore, let $\boldsymbol{\alpha}$ be a row vector with entries $\alpha_i$ for all $i$ that $|\alpha_i|$ is the Euclidean norm of the $i^{\mathrm{th}}$ row $\mathbf{A}_i$ of $\mathbf{A}$. Let $\boldsymbol{\Omega}$ denote the symmetric matrix of the cosines of angles between the rows of $\mathbf{A}$, such that its entries $\Omega_{ij}\triangleq\cos\angle(\mathbf{A}_i,\mathbf{A}_j)$, $i,j\in \{1,\ldots,M\}$, with $|\Omega_{ij}|\leq 1$ for $i\neq j$ and $\Omega_{ij}=1$ for $i=j$. Rewriting \eqref{eq:OrigProHP_approx_1a} with these variables, we have
	\begin{equation}\label{eq:Vari_Mat2Vec1}
	\begin{aligned}
	\max_{\substack{\boldsymbol{\alpha},\boldsymbol{\Omega}}}\quad & \sum_{i=1}^{M}\sum_{j=1}^{M}\frac{1}{2}\Delta_i\Delta_j|\alpha_i \| \alpha_j|\Omega_{ij}\\
	\mathrm{s.t.} \quad 
	&\Omega_{ii}= 1 \\
	&|\Omega_{ij}|\leq 1 \quad i\neq j \in \{1, \ldots, M\} \\
	 & \frac{1}{2}\sum_{i=1}^{M}p_{ki}\alpha_i^2\leq \epsilon_k, \quad k=1,2\\
	\end{aligned}
	\end{equation}
Consider first the optimization over $\boldsymbol{\Omega}$. Since the objective   is linear in $\boldsymbol{\Omega}$ and the feasible region of $\boldsymbol{\Omega}$ is a hypercube, \eqref{eq:Vari_Mat2Vec1} is a linear program whose optimal solution is at one of the extreme points of the hypercube \cite[Theorem 3.5.3]{Nonlinear_Programming}, i.e., the optimal solution $\mathbf{\Omega}^*$ has entries $|\Omega_{i,j}^*|=1$ for all $i,j$. Thus, all the rows of an $\mathbf{A}$ maximizing~\eqref{eq:OrigProHP_approx_1a} are parallel, and therefore, the optimal solution $\mathbf{A}^*$ of~\eqref{eq:OrigProHP_approx_1a} is a rank-1 matrix.

In addition, from the objective function in \eqref{eq:Vari_Mat2Vec1}, if the signs of $\Delta_1\Delta_i$ and $\Delta_1\Delta_j$ are known for any $i,j\in \{2,\ldots, M\}$, the sign of $\Delta_i\Delta_j$ can be determined. Furthermore, maximizing the objective requires that $\Omega_{ij}$ has the same sign as its coefficient $\Delta_i\Delta_j$. Therefore, $\Omega^*_{ij}=\Omega^*_{1i}\Omega^*_{1j}$ for all $i,j\in\{2,\ldots,M\}$, i.e., $\boldsymbol{\Omega}^*$ has only $M-1$ independent entries $\Omega^*_{1j}$, $j\in \{2,\ldots,M\}$ with $\Omega_{ii}^*=1$ for all $i$.

Thus, we see that the optimization in \eqref{eq:Vari_Mat2Vec1} depends on only $M$ values of $|\alpha_i|$, $i\in \{1,\ldots,M\}$, and $M-1$ signs of $\Omega^*_{1i}$, $i\in \{2,\ldots,M\}$. Let $\mathbf{v}$ denote a unit norm vector with no zero entry, 
and $\theta_i^*$ for all $i$ represent the direction of $i^{\mathrm{th}}$ row of $\mathbf{A}$ with respect to $\mathbf{v}$, such that $\Omega^*_{ij}=\theta_i^*\theta_j^*$, and the $i^{\mathrm{th}}$ row of $\mathbf{A}$ can be written as $\mathbf{A}_i=\theta_i^*|\alpha_i|\mathbf{v}=\alpha_i\mathbf{v}$. 
The optimization in \eqref{eq:Vari_Mat2Vec1} can now be written as a function of the vector $\boldsymbol{\alpha}$ as
	\begin{equation}\label{eq:Vari_Mat2Vec2}
	\begin{aligned}
		\max_{\boldsymbol{\alpha}}\quad & \sum_{i=1}^{M}\sum_{j=1}^{M}\frac{1}{2}(\mathbf{p}_1-\mathbf{p}_2)_i(\mathbf{p}_1-\mathbf{p}_2)_j\alpha_i\alpha_j\\
		\mathrm{s.t.} \quad 
		& \frac{1}{2}\sum_{i=1}^{M}p_{1i}\alpha_i^2\leq \epsilon_1\\
		&\frac{1}{2}\sum_{i=1}^{M}p_{2i}\alpha_i^2 \leq \epsilon_2.
	\end{aligned}
	\end{equation}
	The optimal solution $\mathbf{A}^*$ of \eqref{eq:OrigProHP_approx_1a} is then obtained from the $\boldsymbol{\alpha}^*$ optimizing \eqref{eq:Vari_Mat2Vec2} as
\begin{align}
\label{eq:optimal_A}
\mathbf{A}^*=(\boldsymbol{\alpha}^*)^T\mathbf{v}.
\end{align}
Note that the optimal solution in \eqref{eq:Vari_Mat2Vec2} will yield both the magnitude and sign of $\alpha_i$ for all $i$.
The $\mathbf{v}$ in \eqref{eq:optimal_A} can be chosen to satisfy
\begin{align}
\label{eq:optimalA_ORTH}
	\mathbf{A}^*(\sqrt{\mathbf{w}_0})^T=(\boldsymbol{\alpha}^*)^T\mathbf{v}(\sqrt{\mathbf{w}_0})^T=\mathbf{0}.
\end{align}
Thus, solving for $\boldsymbol{\alpha}^*$ in \eqref{eq:OrigProHP_approx_2} leads to $\mathbf{A}^*$ in~\eqref{eq:OrigProHP_approx_1} using \eqref{eq:optimal_A} and \eqref{eq:optimalA_ORTH}.
\end{proof}
\begin{remark}
Knowing the optimal solution $\mathbf{A}^*$ of \eqref{eq:OrigProHP_approx_1}, we can obtain the optimal solution $\boldsymbol{\Theta}^*$ of the E-IT approximation problem \eqref{eq:OrigProHP_approx} as $\boldsymbol{\Theta}^*=\mathbf{A}^*\big[\sqrt{\mathbf{w}_0}\big]$ directly from~\eqref{eq:A_to_Theta}. 
\end{remark}

\subsection{Related Lemmas}
To prove Theorem \ref{Theorem:optimalsol_HPapprox}, we use the simplification in Theorem \ref{Theorem:approximation_simplification} along with two lemmas. The problem in \eqref{eq:OrigProHP_approx_2} maximizes a convex function over a convex feasible region, and thus, it is not a convex program. However, we show how the problem can be related to a convex program and we also obtain a closed-form  solution via the following two lemmas. 

\begin{lemma}\label{lemma_equivalentOS}
		The following convex program completely determines the solutions of \eqref{eq:OrigProHP_approx_2}, 
		\begin{equation}\label{eq:OrigProHP_approx_3simp}
		\begin{aligned}
		\max_{\substack{\boldsymbol{\alpha}}}\quad & 
		\frac{1}{2}\lambda^*\boldsymbol{\alpha}(\mathbf{v}^*)^T\\
		\mathrm{s.t.} \quad &\frac{1}{2}\boldsymbol{\alpha}\big[\mathbf{p}_1\big]\boldsymbol{\alpha}^T\leq \epsilon_1\\
		&\frac{1}{2}\boldsymbol{\alpha}\big[\mathbf{p}_2\big]\boldsymbol{\alpha}^T\leq \epsilon_2,
		\end{aligned}
		\end{equation}
		such that the optimal solutions of \eqref{eq:OrigProHP_approx_2} are $\pm\boldsymbol{\alpha}^*$ where $\boldsymbol{\alpha}^*$ is the optimal solution of \eqref{eq:OrigProHP_approx_3simp}.
\end{lemma}
\begin{proof}
	For the optimization problem in \eqref{eq:OrigProHP_approx_2}, the matrix $(\mathbf{p}_1-\mathbf{p}_2)^T(\mathbf{p}_1-\mathbf{p}_2)$ is rank-1 with eigenvalue $\lambda^*=\|\mathbf{p}_1-\mathbf{p}_2\|^2$ and eigenvector $\mathbf{v}^*=\frac{\mathbf{p}_1-\mathbf{p}_2}{\|\mathbf{p}_1-\mathbf{p}_2\|}$.
	Thus, we have
	\begin{align}
	\frac{1}{2}\boldsymbol{\alpha}(\mathbf{p}_1-\mathbf{p}_2)^T(\mathbf{p}_1-\mathbf{p}_2)\boldsymbol{\alpha}^T\nonumber=\frac{1}{2}\lambda^*(\boldsymbol{\alpha}(\mathbf{v}^*)^T)^2,
	\end{align} 
	leading to the following optimization problem 
	\begin{equation}\label{eq:OrigProHP_approx_3}
	\begin{aligned}
	\max_{\substack{\boldsymbol{\alpha}}}\quad & 
	\frac{1}{2}\lambda^*(\boldsymbol{\alpha}(\mathbf{v}^*)^T)^2\\
	\mathrm{s.t.} \quad &\frac{1}{2}\boldsymbol{\alpha}\big[\mathbf{p}_1\big]\boldsymbol{\alpha}^T \leq \epsilon_1\\
	&\frac{1}{2}\boldsymbol{\alpha}\big[\mathbf{p}_2\big]\boldsymbol{\alpha}^T \leq \epsilon_2
	\end{aligned}
	\end{equation}
	In \eqref{eq:OrigProHP_approx_3simp} and \eqref{eq:OrigProHP_approx_3}, the two objective functions dependent on $\boldsymbol{\alpha}$ in the same manner, and their constraint functions are the same, such that the optimal solution $\boldsymbol{\alpha}^*$ of \eqref{eq:OrigProHP_approx_3simp} optimizes \eqref{eq:OrigProHP_approx_3}. Since the objective and constraint functions of \eqref{eq:OrigProHP_approx_3} are even, $-\boldsymbol{\alpha}^*$ is feasible and gives the optimal value, i.e, $-\boldsymbol{\alpha}^*$ is also the optimal solution of \eqref{eq:OrigProHP_approx_3}. Thus,  when $\boldsymbol{\alpha}^*$ is the optimal solution of \eqref{eq:OrigProHP_approx_3simp}, $\pm\boldsymbol{\alpha}^*$ are the optimal solutions of \eqref{eq:OrigProHP_approx_3}, i.e., the optimal solutions of \eqref{eq:OrigProHP_approx_2}. 
\end{proof}


The optimal solution $\boldsymbol{\alpha}^*$ of \eqref{eq:OrigProHP_approx_3simp} can be evaluated by observing that at $\boldsymbol{\alpha}^*$, either one or both constraints are active. The following lemma summarizes the optimal solution of \eqref{eq:OrigProHP_approx_3simp}. From Lemma \ref{lemma_equivalentOS}, one can then obtain the optimal solution \eqref{eq:OrigProHP_approx_2}.
\begin{lemma}\label{lamma_optimalalpha}
	The optimal solutions of \eqref{eq:OrigProHP_approx_2} are given by: 
	\begin{enumerate}
		\item if only the first constraint is active, i.e., $\epsilon_1$ and $\epsilon_2$ satisfy \eqref{eq:const1_active},
		the optimal solution $\boldsymbol{\alpha}^*$ is \eqref{eq:optalpha_const1};
		\item if only the second constraint is active, i.e., $\epsilon_1$ and $\epsilon_2$ satisfy \eqref{eq:const2_active},
		 the optimal solution $\boldsymbol{\alpha}^*$ is \eqref{eq:optalpha_const2};	
		\item when both constraints are active, the optimal solution $\boldsymbol{\alpha}^*$ is \eqref{eq:optimalalpha}
		with $\eta_1^*>0$ and $\eta_2^*>0$ satisfying \eqref{eq:eta1} and~\eqref{eq:eta2}.
	\end{enumerate}	
\end{lemma}
\begin{proof}
	From Lemma \ref{lemma_equivalentOS}, to find the optimal solutions of \eqref{eq:OrigProHP_approx_2}, it suffices to find the optimal solution to \eqref{eq:OrigProHP_approx_3simp}. In \eqref{eq:OrigProHP_approx_3simp}, the objective function is linear in $\boldsymbol{\alpha}$. Since $\mathbf{p}_1$ and $\mathbf{p}_2$ are interior points of probability simplex, both $\big[\mathbf{p}_1\big]$ and $\big[\mathbf{p}_2\big]$ are positive definite, i.e., the two constraint functions are convex in $\boldsymbol{\alpha}$. Thus, this is a convex program. In addition, there exists $\boldsymbol{\alpha}=\mathbf{0}$ strictly satisfying the two constraints for positive $\epsilon_1$ and $\epsilon_2$, which means the constraints in \eqref{eq:OrigProHP_approx_3simp} satisfy Slater's constraint qualification for positive $\epsilon_1$ and $\epsilon_2$ \cite[Sec.~5.2.3]{boydconvex}. Therefore, the convex program has zero duality gap, and the optimal solutions are the solutions of the following Karush\--Kuhn\--Tucker (KKT) conditions \cite[Sec.~5.5.3]{boydconvex}:
	\begin{subequations}
	      \begin{align}
	            \label{eq:KKTcondition_derivative}
	      		\hspace{-.1in} \nabla\big\{f_0(\boldsymbol{\alpha}^*)\!+\!\eta_1^*\big(\epsilon_1\!-\!
	      		 f_1(\boldsymbol{\alpha}^*)\big)&\!+\!\eta_2^*\big(\epsilon_2\!-\!  f_2(\boldsymbol{\alpha}^*)\big)\big\}  \!=\!0\\
	      		 \label{eq:KKTcondition_loose1}
	      		\eta_1^*\big(\epsilon_1-f_1(\boldsymbol{\alpha}^*)\big) &=0\\
	      		\label{eq:KKTcondition_loose2}
	      		\eta_2^*\big(\epsilon_2-f_2(\boldsymbol{\alpha}^*)\big)&=0\displaybreak[0]\\
	      		\label{eq:KKTcondition_con1}
	      		 f_1(\boldsymbol{\alpha}^*)& \leq \epsilon_1\\
	      		\label{eq:KKTcondition_con2}
	      		 f_2(\boldsymbol{\alpha}^*)& \leq \epsilon_2\\
	      		\label{eq:KKTcondition_dualcon1}
	      		\eta_1^*&\geq 0\\ 
	      		\label{eq:KKTcondition_dualcon2}
	      		 \eta_2^* &\geq 0.
	      \end{align}
	\end{subequations}
	 where $f_0,f_1$, and $f_2$ represent the objective and two constraint functions of \eqref{eq:OrigProHP_approx_3simp}, respectively, $\boldsymbol{\alpha}^*$ is the optimal solution of \eqref{eq:OrigProHP_approx_3simp}, and $\eta^*_1$ and $\eta^*_2$ are the optimal solutions of the dual problem of \eqref{eq:OrigProHP_approx_3simp}. From \eqref{eq:KKTcondition_derivative}, we have
	 \begin{align}
	 \label{eq:optimal_alpha_inproof}
	 \boldsymbol{\alpha}^*=\frac{\lambda^*}{2}\mathbf{v}^*\Big(\eta_1^*\big[\mathbf{p}_1\big]+\eta_2^*\big[\mathbf{p}_2\big]\Big)^{-1}.
	 \end{align}
	 When $\eta^*_1>0$ and $\eta^*_2 = 0$, i.e., the first constraint is active, the optimal solution $\boldsymbol{\alpha}^*$ of \eqref{eq:OrigProHP_approx_3simp} is
	 \begin{align}
	 \label{eq:oneconstraint_proof_alpha1}
	 \boldsymbol{\alpha}^*=\frac{\lambda^*}{2}\mathbf{v}^*\big[(\eta_1^*\mathbf{p}_1)^{-1}\big],
	 \end{align}
	 such that from \eqref{eq:KKTcondition_loose1}, 
	 \begin{align}
	 \label{eq:oneconstraint_proof_eta1}
	 \eta_1^*=\sqrt{\frac{(\lambda^*)^2(\mathbf{v}^*)^T\big[(\mathbf{p}_1)^{-1}\big]\mathbf{v}^*}{8\epsilon_1}}.
	 \end{align}
	 Substituting $\eta_1^*$ from \eqref{eq:oneconstraint_proof_eta1} in \eqref{eq:oneconstraint_proof_alpha1}, we obtain 
	 \begin{align}
	 \label{eq:optalpha_const1_p}
	 \boldsymbol{\alpha}^*=\sqrt{\frac{2\epsilon_1}{\mathbf{v}^*\big[(\mathbf{p}_1)^{-1}\big](\mathbf{v}^*)^T}}\mathbf{v}^*\big[(\mathbf{p}_1)^{-1}\big]
	 \end{align}
	 In addition, for $\eta^*_2 = 0$, \eqref{eq:KKTcondition_con2} is a strict inequality if and only if   $\epsilon_1$ and $\epsilon_2$ satisfy \eqref{eq:const1_active}.
	 
	 When $\eta^*_1=0$ and $\eta^*_2 > 0$, with the same deduction based on \eqref{eq:KKTcondition_loose2} and \eqref{eq:optimal_alpha_inproof}, the optimal solution $\boldsymbol{\alpha}^*$ is
	 \begin{align}
	 	 	 \label{eq:optalpha_const2_p}
	 	 	 \boldsymbol{\alpha}^*=\sqrt{\frac{2\epsilon_2}{\mathbf{v}^*\big[(\mathbf{p}_2)^{-1}\big](\mathbf{v}^*)^T}}\mathbf{v}^*\big[(\mathbf{p}_2)^{-1}\big].
	\end{align} and \eqref{eq:KKTcondition_con1} is a strictly inequality if and only if   $\epsilon_1$ and $\epsilon_2$ satisfy \eqref{eq:const2_active}. 
	
	 When $\eta^*_1>0$ and $\eta^*_2 > 0$, the optimal solution $\boldsymbol{\alpha}^*$ is given by \eqref{eq:optimal_alpha_inproof}. From \eqref{eq:KKTcondition_loose1} and \eqref{eq:KKTcondition_loose2}, we have $f_1(\boldsymbol{\alpha}^*)=\epsilon_1$ and $f_2(\boldsymbol{\alpha}^*)=\epsilon_2$, such that $\eta_1^*$ and $\eta_2^*$ are verified to satisfy \eqref{eq:eta1} and \eqref{eq:eta2}.
	 
	 Finally, since \eqref{eq:optimal_alpha_inproof}, \eqref{eq:optalpha_const1_p} and \eqref{eq:optalpha_const2_p} yield  $\boldsymbol{\alpha}^*$ for \eqref{eq:OrigProHP_approx_3simp}, the optimal solutions for \eqref{eq:OrigProHP_approx_2} are obtained by considering both solutions $\pm\boldsymbol{\alpha}^*$ as proved in the Lemma \ref{lemma_equivalentOS}.
\end{proof}

\subsection{Proof of Theorem \ref{Theorem:optimalsol_HPapprox}}
Theorem \ref{Theorem:optimalsol_HPapprox} presents the optimal \textit{privacy mechanism} $\mathbf{W}'$ for the approximation problem \eqref{eq:OrigProHP_approx}. Its proof is directly from Theorem \ref{Theorem:approximation_simplification}, Lemma \ref{lemma_equivalentOS} and Lemma \ref{lamma_optimalalpha} as follows.  

\begin{proof}
	The optimal \textit{privacy mechanism} $\mathbf{W}'$ is a perturbation of $\mathbf{W}_0$ as $\mathbf{W}'=\mathbf{W}_0+\boldsymbol{\Theta}^*$, where $\boldsymbol{\Theta}^*$ be the optimal solution of the approximation problem \eqref{eq:OrigProHP_approx}. Since $\mathbf{A}^*$ optimizes \eqref{eq:OrigProHP_approx_1}, we have $\boldsymbol{\Theta}^*=\mathbf{A}^*\big[\sqrt{\mathbf{w}_0}\big]$. From Theorem \ref{Theorem:approximation_simplification}, $\mathbf{A}^*$ is given by $\mathbf{A}^*=(\boldsymbol{\alpha}^*)^T\mathbf{v}$, where $\mathbf{v}$ is a unit norm $M$\-- dimensional vector that is orthogonal to $\sqrt{\mathbf{w}_0}$, and $\boldsymbol{\alpha}^*$ is the optimal solution of \eqref{eq:OrigProHP_approx_2} as presented in Lemma \ref{lamma_optimalalpha}. Therefore, the optimal \textit{privacy mechanism} $\mathbf{W}'$ is \eqref{eq:W_generation}.	
\end{proof}

\begin{figure}[t]
	\centering
	\includegraphics[width=3.35in]{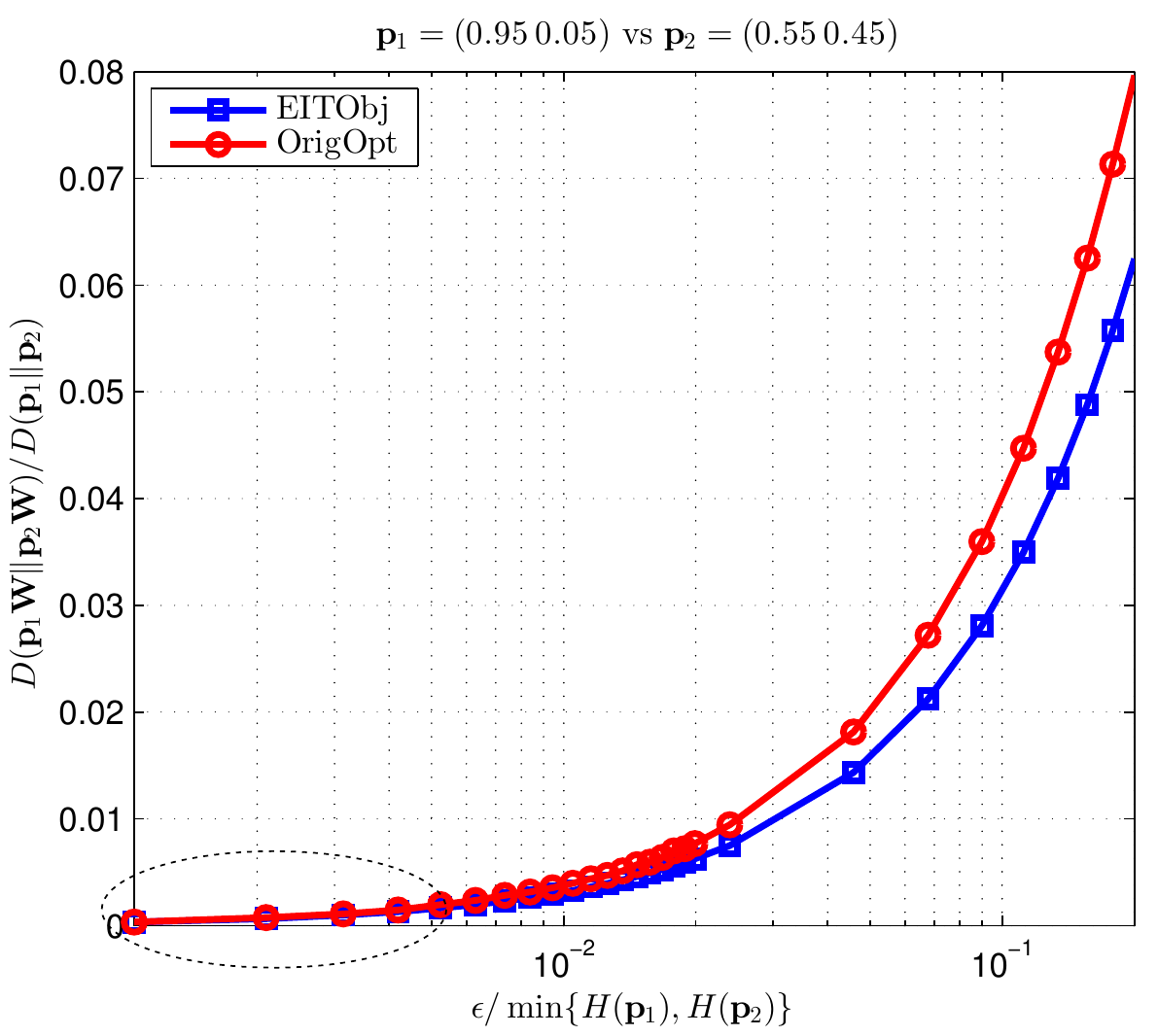}
	\caption{The relative utilities of $\mathbf{W}'$ and $\mathbf{W}^*$ for $\mathbf{p}_1=( 0.95,  0.05)$ and $\mathbf{p}_2=(0.55,0.45)$.}
	\label{subfig:fig1_1}
\end{figure}
\begin{figure}[t]
	\centering
	\includegraphics[width=3.35in]{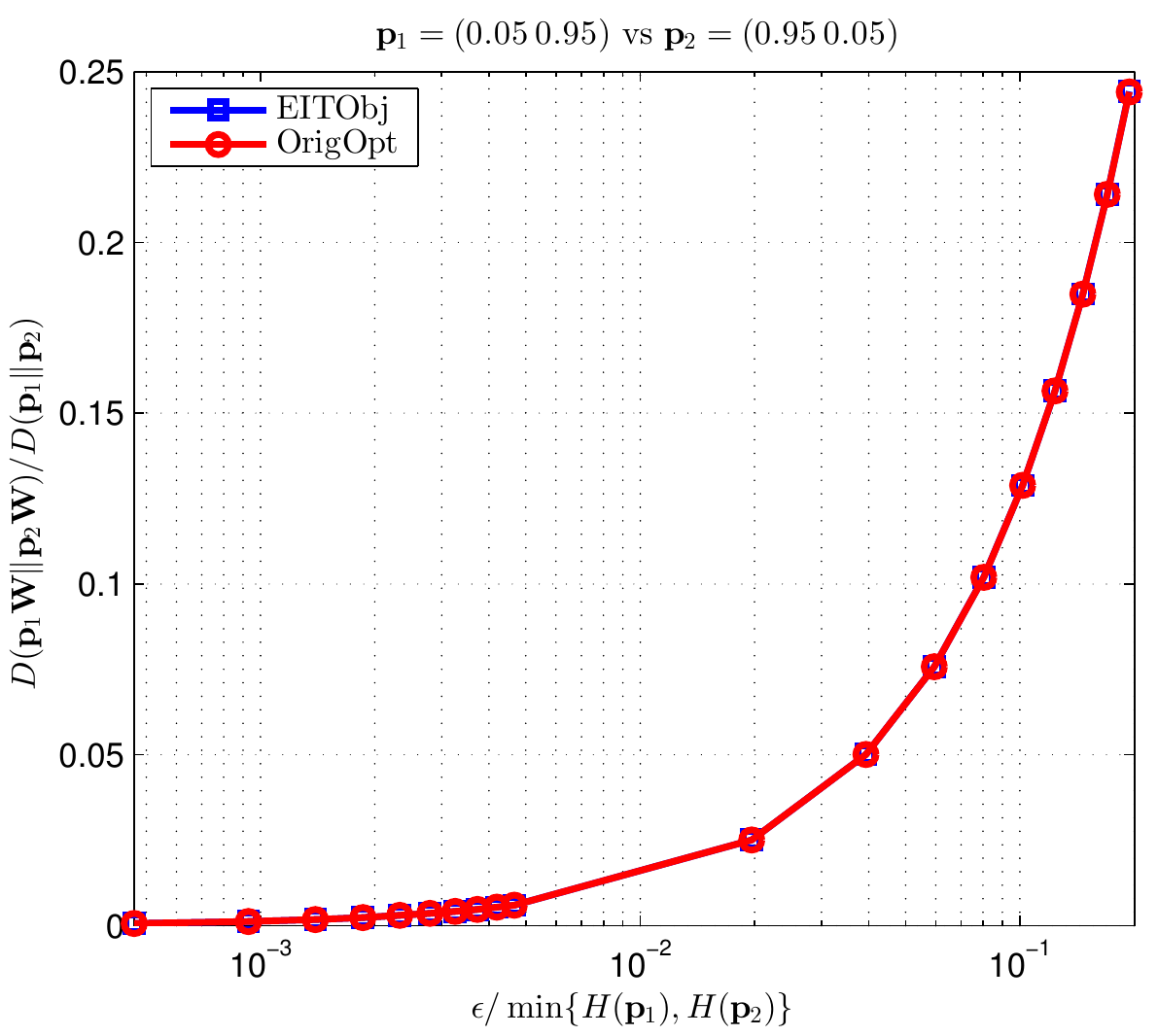}
	\caption{The relative utilities of $\mathbf{W}'$ and $\mathbf{W}^*$ for $\mathbf{p}_1=(0.05,0.95)$ and $\mathbf{p}_2=( 0.95,0.05)$.}
	\label{subfig:fig1_2}
\end{figure}
\section{ILLUSTRATION OF RESULTS}
Theorem \ref{Theorem:optimalsol_HPapprox} provides an approximation, $\mathbf{W}'$ in \eqref{eq:W_generation}, of the optimal \textit{privacy mechanism} for binary hypothesis testing. To evaluate the utility and privacy of $\mathbf{W}'$, we compare the values of $D(\mathbf{p}_1\mathbf{W}' \| \mathbf{p}_2\mathbf{W}')$ and $D(\mathbf{p}_1\mathbf{W}^* \| \mathbf{p}_2\mathbf{W}^*)$, where $\mathbf{W}^*$ is the optimal solution of the original utility-privacy tradeoff problem in \eqref{eq:OrigProHP}, when $I(\mathbf{p}_1,\mathbf{W}^*)$ and $I(\mathbf{p}_2,\mathbf{W}^*)$ are bounded by $\epsilon_1=\epsilon_2=\epsilon=\max\{I(\mathbf{p}_1,\mathbf{W}'),I(\mathbf{p}_2,\mathbf{W}')\}$. 

For the original utility-privacy tradeoff problem in \eqref{eq:OrigProHP}, the number of degrees of freedom in $\mathbf{W}$ is $M(M-1)$ for $M=N$. When $M=3$, finding the optimal \textit{privacy mechanism} $\mathbf{W}^*$ by the exhaustive  search for every single value of  $\epsilon$ is computationally expensive, taking hours on a desktop computer. Therefore, we only consider $M=N=2$. In this case, the probability distributions $\mathbf{p}_1$ and $\mathbf{p}_2$ are Bernoulli. We choose $\mathbf{w}_0$ as the uniform Bernoulli distribution (i.e., $\mathbf{w}_0=(0.5, 0.5)$) because this yielded the best approximations over all choices of $(\mathbf{p}_1,\mathbf{p}_2)$ in our numerical calculations.  The uniform distribution for $\mathbf{w}_0$ yields the best E-IT approximations.  With this choice of $\mathbf{w}_0$, from \eqref{eq:lemma1_v} and \eqref{eq:lemma1_vnorm}, $\mathbf{v}=\pm (\sqrt{0.5} ,-\sqrt{0.5})$.  In the following, we select four pairs of Bernoulli distributions for the two source classes to evaluate the approximation in a high privacy regime where $\epsilon \leq 0.2\min\{H(\mathbf{p}_1), H(\mathbf{p}_2)\}$. 

\begin{figure}[t]
	\centering
	\includegraphics[width=3.35in]{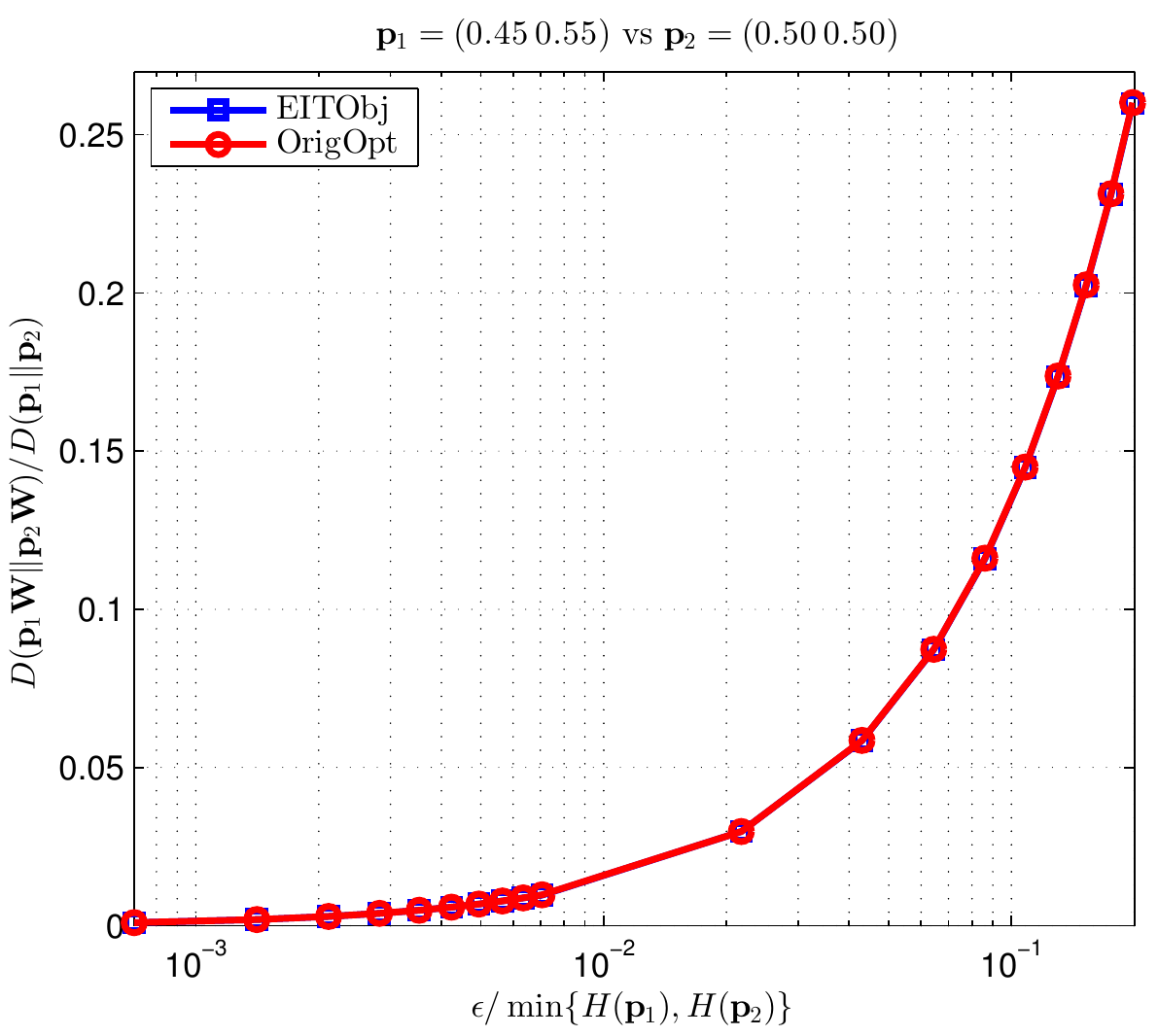}
	\caption{The relative utilities of $\mathbf{W}'$ and $\mathbf{W}^*$ for $\mathbf{p}_1=(0.45,0.55)$ and $\mathbf{p}_2=(0.50, 0.50)$.}
	\label{subfig:fig1_3}
\end{figure}

In Fig.~\ref{subfig:fig1_1}, the pair of   probability distributions for the two source classes are $\mathbf{p}_1=( 0.95,0.05)$ and $\mathbf{p}_2=(0.55,0.45)$. The approximation works well only in the region surrounded by the black-dotted ellipse in the plot, where the statistical leakages $I(\mathbf{p}_k,\mathbf{W}')$, $k\in\{1,2\}$, are less than $0.5\%$ of the minimal information measure of the two sources, namely  $\min\{H(\mathbf{p}_1),H(\mathbf{p}_2)\}$.  Fig.~\ref{subfig:fig1_2} and Fig.~\ref{subfig:fig1_3} show that the approximation performs excellently in the whole regime of interest for the two pairs of probability distributions  $\mathbf{p}_1=(0.05, 0.95)$ and $\mathbf{p}_2=( 0.95,  0.05)$  as well as $\mathbf{p}_1=(0.45, 0.55)$ and $\mathbf{p}_2=(0.50, 0.50)$.  In Fig.~\ref{subfig:fig1_4}, the two source classes are distributed according to $\mathbf{p}_1=(0.05, 0.95)$ and $\mathbf{p}_2=(0.10 , 0.90)$, respectively. The utilities for $\mathbf{W}'$ and $\mathbf{W}^*$ are almost the same only in the region surrounded by the black-dotted ellipse in the plot, where the statistical leakages $I(\mathbf{p}_k,\mathbf{W}')$, $k\in\{1,2\}$, are at most $0.1\%$ of the minimal information measure of the two sources.   

From Figs.~\ref{subfig:fig1_1}--\ref{subfig:fig1_4}, we deduce that for any two source distributions, there exists a high privacy regime where the Euclidean approximation is accurate. However, the specific regime in which the approximation works well differs for different pairs of source distributions. Specifically, when both distributions are close to the uniform or when both distributions are far apart from the uniform as well as each other, the set of leakage values for which the Euclidean approximation works well is larger. For the former, it can be easily seen that the Euclidean approximations of the relative entropy and mutual information are more accurate; for the latter, intuitively, the individual approximation errors  ``cancel out'' so the overall approximation is accurate. 


\section{Discussion and Concluding Remarks}
We have studied the statistical inference problem of binary hypothesis testing under privacy constraints for large datasets. The goal is to understand the guarantees that can be made on the probabilities of error for any data sequence when a mutual information based constraint on the leakage of data from either source classes is restricted. Our model seeks to understand if one can estimate the underlying source class of a given dataset without revealing the respondents of the data, i.e., those whose data is being used for testing. We have shown that the resulting utility-privacy tradeoff problem is one of determining the randomizing privacy mechanism which maximizes the relative entropy (statistical utility) between the output source classes while ensuring that the mutual information based leakages for both source classes are bounded. We have focused on the high privacy regime and developed a Euclidean approximation of the tradeoff problem; for this latter problem, we have shown that the optimal mechanism can be viewed as a perturbation of a near perfect privacy mechanism where the perturbation is computed as a solution to a convex optimization problem.

As is expected of statistical metrics such as relative entropy and mutual information, our results reveal that the randomizing mechanism perturbs the statistical outliers the most in each source class; such a mechanism ensures both utility (predominantly provided by the non-outliers) while preserving privacy of those most vulnerable to inference attacks. Finally, our results also suggest that it is possible to achieve a positive error exponent even in the high privacy regime admittedly with higher sample complexity (since the exponent is small), i.e., likened to the moderate deviations regime.  This work can be extended to study multi\--hypothesis tests, and partial tests with statistical knowledge. One can also study other privacy regimes (low, medium) in which the tradeoff can be approximately quantified.
\begin{figure}[t]
	\centering
	\includegraphics[width=3.35in]{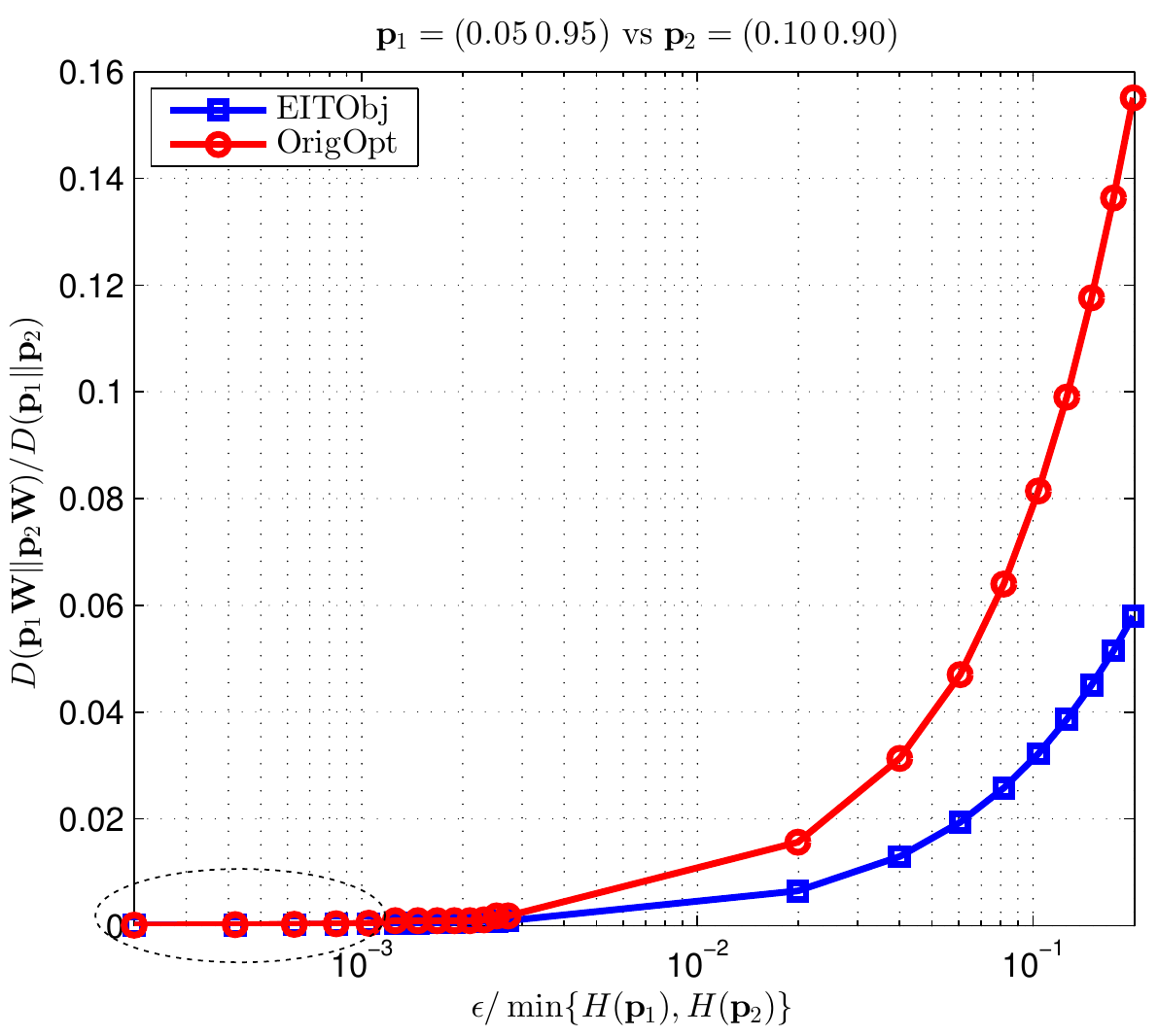}
	\caption{The relative utilities of $\mathbf{W}'$ and $\mathbf{W}^*$ for $\mathbf{p}_1=(0.05 ,0.95)^T$ and $\mathbf{p}_2=(0.10, 0.90)^T$.}
	\label{subfig:fig1_4}
\end{figure}

\addtolength{\textheight}{-12cm}   


\vspace{-.05in}
\bibliographystyle{IEEEtran}
\bibliography{Allerton_2016}

\begin{thebibliography}{10}
\providecommand{\url}[1]{#1}
\csname url@samestyle\endcsname
\providecommand{\newblock}{\relax}
\providecommand{\bibinfo}[2]{#2}
\providecommand{\BIBentrySTDinterwordspacing}{\spaceskip=0pt\relax}
\providecommand{\BIBentryALTinterwordstretchfactor}{4}
\providecommand{\BIBentryALTinterwordspacing}{\spaceskip=\fontdimen2\font plus
\BIBentryALTinterwordstretchfactor\fontdimen3\font minus
  \fontdimen4\font\relax}
\providecommand{\BIBforeignlanguage}[2]{{%
\expandafter\ifx\csname l@#1\endcsname\relax
\typeout{** WARNING: IEEEtran.bst: No hyphenation pattern has been}%
\typeout{** loaded for the language `#1'. Using the pattern for}%
\typeout{** the default language instead.}%
\else
\language=\csname l@#1\endcsname
\fi
#2}}
\providecommand{\BIBdecl}{\relax}
\BIBdecl

\bibitem{EITzheng2008}
S.~Borade and L.~Zheng, ``Euclidean information theory,'' in \emph{2008 IEEE
  International Zurich Seminar on Communications}, 2008.

\bibitem{EIT2015}
S.~Huang, C.~Suh, and L.~Zheng, ``Euclidean information theory of networks,''
  \emph{IEEE Trans. on Inform. Th.}, vol.~61, no.~12, pp. 6795--6814, 2015.

\bibitem{Kairouz2014}
P.~Kairouz, S.~Oh, and P.~Viswanath, ``Extremal mechanisms for local
  differential privacy,'' in \emph{Advances in Neural Information Processing
  Systems}, 2014.

\bibitem{Li_Ochetering}
Z.~Li and T.~J. Oechtering, ``Privacy on hypothesis testing in smart grids,''
  in \emph{IEEE Inform. Th. Workshop (ITW)}, Jeju, South Korea, 2015.

\bibitem{IT_Cover}
T.~M. Cover and J.~A. Thomas, \emph{Elements of Information Theory},
  2nd~ed.\hskip 1em plus 0.5em minus 0.4em\relax Wiley-Interscience, 2006.

\bibitem{MinConcave}
K.~L. Hoffman, ``A method for globally minimizing convex functions over convex
  sets,'' \emph{Mathematical Programming}, vol.~20, pp. 22--31, 1981.

\bibitem{ITSt_Tutorial}
I.~Csisz{\'a}r and P.~C. Shields, \emph{Information theory and statistics: A
  tutorial}.\hskip 1em plus 0.5em minus 0.4em\relax Now Publishers Inc, 2004.

\bibitem{Tan11}
V.~Y.~F. Tan, A.~Anandkumar, L.~Tong, and A.~S. Willsky, ``A large-deviation
  analysis of the maximum-likelihood learning of {Markov} tree structures,''
  \emph{IEEE Trans. on Inform. Th.}, vol.~57, no.~3, pp. 1714--1735, 2011.

\bibitem{Nonlinear_Programming}
M.~S. Bazaraa, H.~D. Sherali, and C.~M. Shetty, \emph{Nonlinear Programming:
  Theory and Algorithms}, 3rd~ed.\hskip 1em plus 0.5em minus 0.4em\relax
  Wiley-Interscience, 2006.

\bibitem{boydconvex}
S.~Boyd and L.~Vandenberghe, \emph{Convex optimization}.\hskip 1em plus 0.5em
  minus 0.4em\relax Cambridge university press, 2014.

\end{thebibliography}

\end{document}